\documentclass{article}
\usepackage{amsmath}
\usepackage{amsfonts}
\usepackage{algorithm,caption}
\usepackage[noend]{algorithmic}

\algsetup{indent=2em}

\usepackage{url}
\newcommand{\qed}{\hfill \ensuremath{\Box}}
\usepackage{chngpage}

\usepackage{fullpage}
\usepackage[affil-it]{authblk}

\usepackage{makeidx}

\newtheorem{lemmx}{Lemma}
\newtheorem{thmx}{Theorem}

\newenvironment{proof}{{\bf Proof:}}{\qed}

\renewcommand{\Pr}{\mbox{\rm \bf{Pr}}}

\begin{document}
\title{Consistent Subset Sampling\thanks{This work is supported by the Danish National Research Foundation under the Sapere Aude program.}}
\date{}
\author{Konstantin Kutzkov\thanks{\tt{kutzkov@gmail.com}}}
\author{Rasmus Pagh\thanks{\tt{pagh@itu.dk}}}
\affil{IT University of Copenhagen, Denmark}
\maketitle

\begin{abstract}
Consistent sampling is a technique for specifying, in small space, a subset $S$ of a potentially large universe $U$ such that the elements in $S$ satisfy a suitably chosen sampling condition. Given a subset $\mathcal{I}\subseteq U$ it should be possible to quickly compute $\mathcal{I}\cap S$, i.e., the elements in $\mathcal{I}$ satisfying the sampling condition. Consistent sampling has important applications in similarity estimation, and estimation of the number of distinct items in a data stream.

In this paper we generalize consistent sampling to the setting where we are interested in sampling size-$k$ subsets occurring in some set in a collection of sets of bounded size $b$, where $k$ is a small integer. This can be done by applying standard consistent sampling to the $k$-subsets of each set, but that approach requires time $\Theta(b^k)$. Using a carefully designed hash function, for a given sampling probability $p \in (0,1]$, we show how to improve the time complexity to $\Theta(b^{\lceil k/2\rceil}\log \log b + pb^k)$ in expectation, while maintaining strong concentration bounds for the sample. The space usage of our method is $\Theta(b^{\lceil k/4\rceil})$. 

We demonstrate the utility of our technique by applying it to several well-studied data mining problems. We show how to efficiently estimate the number of frequent $k$-itemsets in a stream of transactions and the number of bipartite cliques in a graph given as incidence stream. Further, building upon a recent work by Campagna et al., we show that our approach can be applied to frequent itemset mining in a parallel or distributed setting. We also present applications in graph stream mining.

\end{abstract}

\section{Introduction}
\label{intro}
{\em Consistent sampling\/} is an important technique for constructing randomized sketches (or ``summaries'') of large data sets. The basic idea is to decide whether to sample an element $x$ depending on whether a certain sampling condition is satisfied. Usually, consistent sampling is implemented using suitably defined hash functions and $x$ is sampled if its hash value $h(x)$ is below some threshold. If $x$ is encountered several times, it is therefore either {\em never\/} sampled or {\em always\/} sampled. The set of items to sample is described by the definition of the hash function, which is typically small.

Consistent sampling comes in two basic variations: 
In one variation (sometimes referred to as {\em subsampling}) there is a fixed sampling probability~$p \in (0,1)$, and elements in a set must be sampled with this probability. In the alternative model the sample size is fixed, and the sampling probability must be scaled to achieve the desired sample size.

Depending on the strength of the hash function used, the sample will exhibit many of the  properties of a random sample (see e.g.~\cite{minwise,indyk-minwise}).
One of the most famous applications of consistent sampling~\cite{Broderetal97} is estimating the {\em Jaccard similarity} of two sets by the similarity of consistent samples, using the same hash function.
Another well-known application is reducing the number of distinct items considered to $\Theta(1/\varepsilon^2)$ in order to make an $(1\pm \varepsilon)$-approximation of the total number of distinct items (see~\cite{opt_dist_elmts} for the state-of-the-art result).

In this paper we consider consistent sampling of certain {\em implicitly defined\/} sets. That is, we sample from a set much larger than the size of the explicitly given database. Our main focus is on streams of sets, where we want to sample subsets of the sets in the stream.  Figure~\ref{fig:example} shows an example result of such sampling.

\begin{figure}[!h]
{\footnotesize\em
\begin{tabular}{p{.55\linewidth}p{.5\linewidth}}
	btw honorable arena constructive\newline
	geps honorable chocola alexander\newline
	geps compared alarm joseph\newline
	granholm honorable drama bother\newline
	granholm flight globe accuracy\newline
	abandoned career crooks abortions\newline
	generation contract attract fitzpatrick\newline
	foolish honorable briefly coattails\newline
	foolish hiring believes fans\newline
	examples internet bipartisan annual\newline
	gasoline corruption criminals connections\newline
	gasoline guilty iran jennifer\newline
	gasoline fewer dinner analyses\newline
	advertisements assertions dec influential\newline
	advertisements jennifer approval injured\newline
	allegation cuts breaking intentions\newline
	ironically dhinmi developed identified\newline
	checks conclusions buildings combined\newline
	handle capitol design israel\newline
	board damaged formed falls\newline
	assertions conventions declined disagreement\newline
	editor judges household district\newline
	editor christians caution castro\newline
	decline hiring gear heartland\newline
	cuts introduced coattails collapsed\newline
	conventions gray circles definition\newline
	grief abuse dkosopedia buck\newline
	grief constructive considered definition\newline
	ending illness hostile johnson\newline
	demonstrate delegation coattails await
&
	buildings employed collapsed employers\newline
	exciting combined accuracy huh\newline
	accidentally adopt drops keever\newline
	competing eye introduced began\newline
	abraham administrations began governor\newline
	immigrants ambush confederate chafee\newline
	immigrants attract abortions hacks\newline
	gubenatorial fire blogpac fighting\newline
	dcalif entered includes governor\newline
	district includes grown blitzer\newline
	collins contenders enjoy connections\newline
	household gear arena atlarge\newline
	christians israel exam believes\newline
	israel club dust governor\newline
	israel alqaida anticipated bear\newline
	dec ambush believes cancelled\newline
	club bradnickel abuse damaged\newline
	bipartisan foe commenting emailed\newline
	influential bandwagon dubya adopt\newline
	important gadflyer envoy abortions\newline
	devastating flight death fulltime\newline
	hmm jennifer assistant falls\newline
	alarm discusses investigations earn\newline
	corner bother attract intentions\newline
	corner california attract dkosopedia\newline
	andrew fighting dkosopedia elect\newline
	contribution grown envoy johnson\newline
	assistant abuse believes changing\newline
	documentation fulltime began dubya\newline
	craft ignorant alexander headquarters	
\end{tabular}
}
\caption{Random sample of 4-word sets of frequency at least 1\% on the Kos blog (data set from the UCI Machine Learning Repository). There are 491134490 such sets, and the above represents a sample of a fraction $10^{-7}$ of these. Our algorithm samples sets of size 4 by identifying disjoint subsets of size 2 that collide under a suitably chosen hash function, resulting in pairwise independent samples. We use a technique from subset-sum algorithms to identify collisions using space linear in the number of distinct items, rather than the number of distinct pairs. In the Kos data set there are 6906 distinct words (excluding stop words), so there is a great difference between linear and quadratic time.
}\label{fig:example}
\end{figure}


We demonstrate the usability of our technique by designing new algorithms for several well-studied counting problems in the streaming setting. We present the first nontrivial algorithm for the problem of estimating the number of frequent $k$-itemsets with rigorously understood complexity and error guarantee and also give a new algorithm for counting bipartite cliques in a graph given as an incidence stream. Also, using a technique presented in~\cite{pairse}, we show that our approach can be easily parallelized and applied to frequent itemset mining algorithms based on hashing~\cite{count_sketch,count_min}.
\section{Preliminaries}

\paragraph{\bf Notation.}
Let $\mathcal{C} = T_1,.., T_m$ be a collection of $m$ subsets of a ground set $\mathcal{I}$, $T_j \subseteq \mathcal{I}$, where $\mathcal{I}=\{1,\dots,n\}$ is a set of {\em elements}. The sets $T_j$ each contain at most $b$ elements, i.e., $|T_j| \le b$, and in the following are called {\em $b$-sets}.
Let further $S \subseteq \mathcal{I}$ be a given subset. If $|S| = k$ we call $S$ a {\em $k$-subset}. We assume that the $b$-sets are explicitly given as input while a $k$-subset can be any subset of $\mathcal{I}$ of cardinality $k$. In particular, a $b$-set with $b$ elements contains $b \choose k$ distinct $k$-subsets for $k \le b$. The {\em frequency} of a given  $k$-subset is the number of $b$-sets containing it. 

In order to simplify the presentation, we assume a lexicographic order on the elements in $\mathcal{I}$ and a unique representation of subsets as ordered vectors of elements. However, we will continue to use standard set operators to express computations on these vectors.
In our algorithm we will consider only lexicographically ordered $k$-subsets. 
For two subsets $I_1, I_2$ we write $I_1 < I_2$ iff $i_1 < i_2$ $\forall i_1 \in I_1, i_2 \in I_2$.


The set of $k$-subsets of $\mathcal{I}$ is written as $\mathcal{I}^k$ and similarly, for a given $b$-set $T_j$, we write $T^k_j$ for the family of $k$-subsets occurring in $T_j$. A family of $k$-subsets
$\mathcal{S} \subset \mathcal{I}^k$ is called a {\em consistent sample} for a given sampling condition $P$ if for each $b$-set $T_i$ the set $\mathcal{S} \cap T_i^k$ is sampled, i.e., all elements satisfying the sampling condition $P$ that occur in $T_i$ are sampled. The sampling condition $P$ will be defined later.

Let $[q]$ denote the set $\{0,\dots, q-1\}$ for $q \in \mathbb{N}$.
A hash function $h: \mathcal{E} \rightarrow [q]$ is $t$-wise independent iff $\Pr[h(e_1) = c_1 \wedge h(e_2) = c_2 \wedge \dots \wedge h(e_t) = c_t] = 
q^{-t}$
for distinct elements $e_i \in \mathcal{E}$, $1 \leq i \leq t$, and $c_i \in [q]$. We denote by $p=1/q$ the sampling probability we use in our algorithm. Throughout the paper we will often exchange $p$ and $1/q$.

We assume the standard computation model and further we assume that one element of $\mathcal{I}$ can be written in one machine word.

\paragraph{Probability theory.}
We assume that the reader is familiar with basic definitions from probability theory. In the analysis of our algorithms we will use these inequalities:

\begin{itemize}
\item {\em Markov's inequality} Let $X$ be a random variable. Then for every $\lambda > 1$:
\begin{eqnarray}
\Pr[X \geq \lambda] \leq \frac{\mathbb{E}[X]}{\lambda}
\end{eqnarray} 

\item {\em Chebyshev's inequality.} Let $X$ be a random variable and $\lambda>0$. Then

\begin{eqnarray} \label{chebyshev}
\Pr[|X-\mathbb{E}[X]| \geq \lambda] \leq \frac{\mathbb{V}[X]}{\lambda^2}
\end{eqnarray}

\item {\em Chernoff's inequality.} We will use the following form of the inequality:

Let $X_1, \ldots, X_\ell$ be $\ell$ independent identically distributed Bernoulli random variables and $\mathbb{E}[X_i] = \mu$. Then for any $\varepsilon > 0$ we have 
\begin{eqnarray} \label{chernoff}
\Pr[|\frac{1}{\ell}\sum_{i=1}^\ell X_i - \mu| > \varepsilon\mu] \leq 2e^{-\varepsilon^2 \mu \ell/2}
\end{eqnarray}
\end{itemize}

\paragraph{\bf Example.} In order to simplify further reading let us consider a concrete data mining problem. Let $\mathcal{T}$ be a stream of $m$ transactions $T_1, T_2, \dots, T_m$ each of size $b$. Each such transaction is a subset of the ground set of items $\mathcal{I}$. We consider the problem of finding the set of frequent $k$-itemsets, i.e., subsets of $k$ items occurring in at least $t$ transactions for a user-defined $t\leq m$.   As a concrete example consider a supermarket. The set of items are all offered goods and transactions are customers baskets. Frequent 2-itemsets will provide knowledge about goods that are frequently bought together.

The problem can be phrased in terms of the above described abstraction by associating transactions with $b$-sets and $k$-itemsets with $k$-subsets. Assume we want to sample 2-itemsets. A consistent sample can be described as follows: for a hash function $h: \mathcal{I} \rightarrow [q]$ we define $S$ to be the set of all 2-subsets $(i, j)$ such that $h(i) + h(j) = 0 \text{ mod } q$. In each $b$-set we can then generate all $b \choose 2$ 2-subsets and check which of them satisfy the so defined sampling condition. For a suitably defined hash function, one can show that resulting sample is ``random enough" and can provide important information about the data, for example, we can use it to estimate the number of 2-itemsets occurring above a certain number of times. 
%
%
%
%
%
%
%


%
%
%
%
%
%

\section{Our contribution}

\subsection{Time-space trade-offs revisited.}

Streaming algorithms have traditionally been mainly concerned with space usage.
An algorithm with a superior space usage, for example polylogarithmic, has been considered superior to an algorithm using more space but less computation time.
We would like to challenge this view, especially for time complexities that are in the polynomial (rather than polylogarithmic) range.
The purpose of a scalable algorithm is to allow the largest possible problem sizes to be handled (in terms of relevant problem parameters).
A streaming algorithm may fail either because the processing time is too high, or because it uses more space than what is available.
Typically, streaming algorithms should work in space that is small enough to fit in fast cache memory, but there is no real advantage to using only 10\% of the cache.
Looking at high-end processors over the last 20 years, see for example~\url{http://en.wikipedia.org/wiki/Comparison_of_Intel_Processors}, reveals that the largest system cache capacity and the number of instructions per second have developed rather similarly (with the doubling time for space being about 25\% larger than the doubling time for number of instructions).
Assuming that this trend continues, a future processor with $x$ times more processing power will have about $x^{0.8}$ times larger cache.
So informally, whenever we have $S=o(T^{0.8})$ for an algorithm using time $T$ and space $S$ the space will not be the asymptotic bottleneck.

\subsection{Main Result.}

In this paper we consider consistent sampling of certain implicitly defined sets, focusing on size-$k$ subsets in a collection of $b$-sets. The sampling is consistent in the sense that each occurrence of a $k$-subset satisfying the sampling condition is recorded in the sample.

\begin{thmx} \label{thm:main}
For each integer $k\geq 2$ there is an algorithm computing a consistent, pairwise independent sample of $k$-subsets from a given $b$-set in expected time $O(b^{\lceil k/2 \rceil}\log \log b + pb^k)$ and space $O(b^{\lceil k/4 \rceil})$ for a given sampling probability $p$, such that $1/p = O(b^k)$ and $p$ can be described in one word. An element of the sample is specified in $O(k)$ words.
\end{thmx}

Note that for the space complexity we do not consider the size of the computed sample. We will do this when presenting concrete applications of our approach.

For low sampling rates our method, which is based on hash collisions among $k/2$-subsets, is a quadratic improvement in running time compared to the na\"ive method that iterates through all $k$-subsets in a given $b$-set. 
In addition, we obtain a quadratic improvement in space usage compared to the direct application of the hashing idea.
Storing a single $2k$-wise independent hash function suffices to specify a sample, where every pair of $k$-subsets are sampled independently.

An important consequence of our consistent sampling algorithm is that it can be applied to $b$-sets revealed one at a time, thus it is well-suited for streaming problems. 



%
%
%
%
%
%

\section{Our approach} 

\subsection{Intuition}


A na\"ive consistent sampling approach works as follows: Define a pairwise independent hash function $h:\mathcal{I}^k \rightarrow [q]$, for a given $b$-set $T$ generate all $b \choose k$ $k$-subsets $I_k \in T^k$ and sample a subset $I_k$ iff $h(I_k) = 0$. Clearly, to decide which $I_k$ are sampled the running time is $O(b^k)$ and the space is $O(b)$ since the space needed for the description of the hash function for reasonably small sampling probability $p$ is negligible. A natural question is whether a better time complexity is possible.

Our idea is instead of explicitly considering all $k$-subsets occurring in a given $b$-set, to hash all elements to a value in $[q]$, $q = \lceil 1/p \rceil$ for a given sampling probability $p$. We show that the sampling of $k$-subsets is {\em pairwise independent\/} and for many concrete applications this is sufficient to consider the resulting sample ``random enough".

The construction of the hash function is at the heart of our algorithm and allows us to exploit several tricks in order to improve the running time. Let us for simplicity assume $k$ is even. Then we sample a given $k$-subset if the sum (mod $q$) of the hash values of its first $k/2$ elements equals the sum of the hash values of its last $k/2$ elements modulo $q$. The simple idea is to sort all $k/2$-subsets according to hash value and then look for collisions. Using a technique similar to the one of Schroeppel and Shamir for the knapsack problem~\cite{schrsham}, we show how by a clever use of priority queues one can design an algorithm with much better time complexity than the na\"ive method and quadratic improvement in the space complexity of the sorting approach. 



\subsection{The hash function}

Now we explain how we sample a given $k$-subset.
Let $h:\mathcal{I} \rightarrow [q]$ be a $2k$-wise independent hash function, $k \geq 2$. It is well-known, see for example ~\cite{hashing}, that such a function can be described in $O(k)$ words for a reasonable sampling probability, i.e., a sampling probability that can be described in one machine word.
 
We take a $k$-subset $(a_1, \dots, a_{\lfloor k/2 \rfloor}, a_{\lfloor k/2 \rfloor + 1}, \dots, a_k)$ in the sample iff $(h(a_1) + \dots + h(a_{\lfloor k/2 \rfloor})) \mbox{ mod } q = (h(a_{\lfloor k/2 \rfloor + 1})+ \dots + h(a_k)) \mbox{ mod } q$. Note that we have assumed a unique representation of subsets as sorted vectors and thus the sampling condition is uniquely defined.

For a given $k$-subset $I=(a_i, a_{i+1} \dots, \dots, a_{i+k-1})$, $i \in \mathcal{I}$, we denote by $h(a_i, a_{i+1} \dots, \dots, a_{i+k-1}))$  the value $(h(a_1) + h(a_{i+1}) + \dots + h(a_{i+k-1})) \mbox{ mod } q$.
We define the random variable $X_I$ to indicate whether a given $k$-subset $I = (a_1, \dots, a_k)$ will be considered for sampling:

\medskip

\begin{equation*}
X_{I} = 
\begin{cases}
 1, & \text{if $h(a_1\dots a_{\lfloor k/2 \rfloor}) = h(a_{\lfloor k/2 \rfloor + 1}\dots a_k)$,}
\\0, & \text{otherwise}
\end{cases}
\end{equation*}

\medskip

The following lemmas allow us to assume that from our sample we can obtain a reliable estimate with high probability:

\begin{lemmx} \label{hash_lemma}
Let $I$ be a $t$-subset with $t\leq k$. Then for a given $r \in [q]$, $\Pr[h(I) = r] = {1}/{q}$.
\end{lemmx}

\begin{proof} 

Since $h$ is $2k$-wise independent and uniform each of the $t \le k$ distinct elements is hashed to a value between $0$ and $q-1$ uniformly and independently from  the remaining $t-1$ elements. Thus, the sum (mod $q$) of the hash values of $I$'s $t$ elements is equal with probability $1/q$ to $r$. 
\end{proof}

\begin{lemmx} \label{sampling_lemma}
For a given $k$-subset $I$, $\Pr[X_I=1] = 1/q$.
\end{lemmx}

\begin{proof} 

Let $I = I_l \cup I_r$ with $|I_l| = \lfloor k/2 \rfloor$ and $|I_r| = \lceil k/2 \rceil$.  The hash value of each subset is uniquely defined, $h$ is $2k$-wise independent, and together with the result of Lemma~\ref{hash_lemma} we have $\Pr[h(I_l) = h(I_r) = r] = 1/q^2$ for a particular $r \in [q]$. Thus, we have $\Pr[h(I_l) = h(I_r) = 0 \vee  \dots \vee h(I_l) = h(I_r) = q-1] = \sum_{i=0}^{q-1} \Pr[h(I_l)=h(I_r) = i] = 1/q$. 
\end{proof}

\begin{lemmx} \label{independence_lemma}
Let $I_1$ and $I_2$ be two distinct $k$-subsets. Then the random variables $X_{I_1}$ and $X_{I_2}$ are independent.
\end{lemmx}

\begin{proof} 

We show that $\Pr[X_{I_1} = 1 \wedge X_{I_2} = 1] = \Pr[X_{I_1}=1]\Pr[X_{I_2}=1] = 1/q^2$ for arbitrarily chosen $k$-subsets $I_1, I_2$. This will imply pairwise independence on the events that two given $k$-subsets are sampled since for a given $k$-subset $I$, $\Pr[X_I = 1] = 1/q$ as shown in Lemma \ref{hash_lemma}.

Let $I_1 = I_{1l} \cup I_{1r}$ and $I_2 = I_{2l} \cup I_{2r}$ with $|I_{il}| = \lfloor k/2 \rfloor$ and $|I_{ir}| = \lceil k/2 \rceil$. Let us assume without loss of generality that $h(I_{1l}) = r_1$ and $h(I_{2l}) = r_2$ for some $r_i \in [q]$. As shown in the previous lemmas for fixed $r_1$ and $r_2$, $\Pr[h(I_{1r}) = r_1] = \Pr[h(I_{2r}) = r_2] =1/q$. Since $h$ is $2k$-wise independent, all elements in $I_{1l} \cup I_{1r} \cup I_{2l} \cup I_{2r}$ are hashed independently of each other. Thus, it is easy to see that the event we hash $I_{2r}$ to $r_2$ is independent from the event that we have hashed $I_{1r}$ to $r_1$, thus the statement follows. 
\end{proof}

The above lemmas imply that our sampling will be uniform and pairwise independent.

\subsection{The algorithm} \label{sec:algorithm} 

\begin{figure}  
\vspace{-15mm}
\renewcommand{\thealgorithm}{}
{\sc ConsistentSubsetSampling}
\begin{algorithmic}[0]
\medskip
\REQUIRE $b$-set $T \subset \mathcal{I}$, a $2k$-wise independent $h: \mathcal{I} \rightarrow [q]$
\STATE Let $H = T^{k/4}$ be the $k/4$-subsets occurring in $T$.
\STATE Sort all $k/4$-subsets from $H$ in a circular list $L$ according to hash value.
\STATE Build a priority queue $P$ with $k/2$-subsets $I = I_H \cup I_L$ according to hash value, for $I_H \in H$, $I_L \in L$.
\FOR{$i \in [q]$}
\STATE $T^{k/2}_i = $ {\sc OutputNext($P, L, i$)}
\STATE Generate all $k$-subsets from $T^{k/2}_i$ satisfying the sampling condition (and consisting of $k$ different elements).
\ENDFOR
\end{algorithmic}


\bigskip
{\sc OutputNext}
\begin{algorithmic}[0]
\REQUIRE a circular list $L$, a priority queue $P$ of $k/2$-subsets $I = (I_H \cup I_L)$ compared by hash value $h(I)$,  $i \in \mathbb{N}$
\WHILE{there is $k/2$-subset with a hash value $i$}
\STATE Output the next $k/2$-subset $I = (I_H \cup I_L)$ from $P$.
\IF{ $\text{cnt}(I_H) < L.\text{length}$}
\STATE Replace $I$ by $I_H \cup I'_L$ in $P$ where $I'_L$ is the $k/4$-subset following $I_L$ in $L$.
\STATE Update the hash value of $I_H \cup I'_L$ and restore the PQ invariant.
\STATE $\text{cnt}(I_H)\text{++}$.
\ELSE 
\STATE Remove $I = (I_H \cup I_L)$ from $P$ and restore the PQ invariant.
\ENDIF
\ENDWHILE
\end{algorithmic}
\caption{A high-level pseudocode description of the algorithm. For simplicity we assume that $k$ is a multiple of 4. The letter $H$ stands for ``head", these are the $k/4$-subsets that will constitute the first half of $k/2$-subsets in $P$. We will always update the second half with a $k/4$-subset from $L$.} \label{fig:css}
\end{figure}

A pseudocode description of our algorithm is given in Figure~\ref{fig:css}. We explain how the algorithm works with a simple example. Assume we want to sample 8-subsets from a $b$-set $(a_1,\ldots,a_b)$ with $b>8$. We want to find all 8-subsets $(a_1,\dots,a_8)$ for which it holds $h(a_1,\dots,a_4)=h(a_5,\dots,a_8)$. As discussed, we assume a lexicographic order on the elements in $\mathcal{I}$ and we further assume $b$-sets are sorted according to this total order. The assumption can be removed by preprocessing and sorting the input. Since $\mathcal{I}$ is discrete, one can assume that each $b$-set can be sorted by the Han-Thorup algorithm in $O(b\sqrt{\log \log b})$ expected time and space $O(b)$~\cite{han_thorup} (for the general case of sampling $k$-subsets even for $k=2$ this will not dominate the complexity claimed in Theorem~\ref{thm:main}).  In the following we assume the elements in each $b$-set are sorted.  Recall we have assumed a total order on subsets, and all subsets we consider are sorted according to this total order. We will also consider only sorted subsets for sampling. 

We simulate a sorting algorithm in order to find all $4$-subsets with equal hash values. Let the set of $2$-subsets be $H$.  First, in {\sc ConsistentSubsetSampling} we generate all $b \choose 2$ 2-subsets and sort them according to their hash value in a circular list $L$ guaranteeing access in expected constant time. We also build a priority queue $P$ containing $b \choose 2$ $4$-subsets as follows: For each 2-subset $(a_i, a_j) \in H$ we find the 2-subset $(a_k, a_\ell) \in L$ such that $h(a_i,a_j, a_k, a_\ell)$ is minimized and keep track of the position of $(a_k, a_\ell)$ in $L$. Then we successively output all 4-subsets sorted according to their hash value from the priority queue by calling {\sc OutputNext}. For a given collection of 4-subsets with the same hash value we generate all valid 8-subsets, i.e., we find all combinations yielding lexicographically ordered 8-subsets. Note that the ``head" 2-subsets from $H$ are never changed while we only update the ``tail" 2-subsets with new 2-subsets from $L$. During the process we also check whether all elements in the newly created 4-subsets are different. 

In {\sc OutputNext} we simulate a heapsort-like algorithm for 4-subsets. We do not keep explicitly all 4-subsets in $P$ but at most $b \choose 2$ 4-subsets at a time. Once we output a given 4-subset $(a_i,a_j, a_k, a_\ell)$ from $P$, we replace it with $(a_i,a_j, a'_k, a'_\ell)$ where $(a'_k, a'_\ell)$ is the 2-subset in $L$ following $(a_k, a_\ell)$. We also keep track whether we have not already traversed $L$ for each 2-subset in $H$. If this is the case, we remove the 4-subset $(a_i,a_j, a_k, a_\ell)$ from $P$ and the number of recorder entries in $P$ is decreased by 1. At the end we update $P$ and maintain the priority queue invariant.

In the following lemmas we will prove the correctness of the algorithm for general $k$ and will analyze its running time. This will yield our main Theorem~\ref{thm:main}. 

\begin{lemmx}
For $k \ge 4$ with $k \text{ mod } 4 = 0$ {\sc Consistent Subset Sampling} outputs the $k/2$-subsets from a given $b$-set in sorted order according to their hash value in expected time $O(b^{k/2} \log \log b)$ and space $O(b^{k/4})$. 
\end{lemmx}
\begin{proof}
Let $T$ be the given $b$-subset and $\mathcal{S} = I_1,\ldots,I_{b \choose k/2}$ be the $k/2$-subsets occurring in $T$ sorted according to hash value.
For correctness we first show that the following invariant holds: After the $j$ smallest $k/2$-subsets have been output from $P$, $P$ contains the $(j+1)$th smallest $k/2$-subset in $S$, ties resolved arbitrarily. For $j=1$ the statement holds by construction. Assume now that it holds for some $j \geq 1$ and we output the $j$th smallest $k/2$-subset $I = I_H \cup I_L$ for $k/4$-subsets $I_H$ and $I_L$. We then replace it by $I' = I_H \cup I'_L$ where $I'_L$ is the $k/4$-subset in $L$ following $I_L$, or, if $L$ has been already traversed, remove $I$ from $P$. If $P$ contains the $(j+1)$th smallest $k/2$-subset, then the invariant holds. Otherwise, we show that it must be that the $(j+1)$th smallest $k/2$-subset is $I'$. Since $L$ is sorted and we traverse it in increasing order, the $k/2$-subsets $I_H \cup I_L$ with a fixed head $I_H$ that remain to be considered have all a bigger hash value than $I$. The same reasoning applies to all other $k/2$-subsets in $P$, and since no of them is the $(j+1)$-th smallest $k/2$-subset, the only possibility is that indeed $I'$ is the $(j+1)$-th smallest $k/2$-subset.

In $L$ we need to explicitly store $O(b^{k/4})$ subsets. Clearly, we can assume that we access the elements in $L$ in constant time.
The time and space complexity depend on how we implement the priority queue $P$. We observe that for a hash function range in $b^{O(1)}$ the keys on which we compare the 2-subsets are from a universe of size $b^{O(1)}$.  Thus, we can implement $P$ as a $y$-fast trie~\cite{y_fast_trie} in $O(b^{k/4})$ space supporting updates in $O(\log \log b)$ time. This yields the claimed bounds. 

\end{proof}

Note however, that the number of $k/2$-subsets with the same hash value might be $\omega(b^{k/4})$. We next guarantee that the worst case space usage is $O(b^{k/4})$. 

\begin{lemmx} \label{lem:gener}
For a given $r \in [q]$, $k \text{ mod } 4 = 0$, and sampling probability $p \in (0,1]$,  we generate all $k$-subsets from a set of $k/2$-subsets with a hash value $r$ that satisfy the sampling condition  in expected time $O(p^2 b^k)$ and space $O(b^{k/4})$. 
\end{lemmx}

\begin{proof}

We use the following implicit representation of $k/2$-subsets with the same hash value. For a given $k/4$-subset $I_P$ the $k/4$-subsets $I_L$ in $L$ occurring in $k/2$-subsets $I_P \cup I_L$ with the same hash value are contained in a subsequence of $L$.  Therefore, instead of explicitly storing all $k/2$-subsets, for each $k/4$-subset $I_P$ we store two indices $i$ and $j$ indicating that $h(I_P \cup L[k]) = r$ for $i \le k \le j$.  Clearly, this guarantees a space usage of $O(b^{k/4})$. 

We expect $p b^{k/2}$ $k/2$-subsets to have hash value $r$, thus the number of $k$-subsets that will satisfy the sampling condition is $O(p^2b^k)$.
\end{proof}

The above lemmas prove Theorem~\ref{thm:main} for the case $k \text{ mod } 4 = 0$. One generalizes to arbitrary $k\ge 4$ in the following way: 

For even $k$ with $k \mbox{ mod } 4 =2$, meaning that $k/2$ is odd, it is easy to see that we need a circular list with all $\lfloor k/4 \rfloor$-subsets but the priority queue will contain $b \choose \lceil k/4 \rceil$ pairs of $k/2$-subsets (which are concatenations of $\lceil k/4 \rceil$-subsets and $\lfloor k/4 \rfloor$-subsets). For odd $k$ we want to sample all $k$-subsets for which the sum of the hash values of the first $\lfloor k/2 \rfloor$ elements equals the sum of the hash values of the last $\lceil k/2 \rceil$ elements. We can run two copies of {\sc OutputNext} in parallel, one will output the $\lceil k/2 \rceil$-subsets with hash value $r$ and the other one the $\lfloor k/2 \rfloor$-subsets with hash value $r$ for all $r \in [q]$. Then we can generate all $k$-subsets satisfying the sampling condition as outlined in Lemma~\ref{lem:gener} with the only difference that we will combine $\lceil k/2 \rceil$-subsets with $\lfloor k/2 \rfloor$-subsets output by each copy of {\sc OutputNext}. Clearly, the space complexity is bounded by $O(b^{\lceil k/4 \rceil})$ and the expected running time is $O(b^{\lceil k/2 \rceil} + pb^k)$.
This completes the proof of Theorem~\ref{thm:main}.

\paragraph{\bf A time-space trade-off.}

A better space complexity can be achieved by increasing the running time. The following theorem generalizes our result.

\begin{thmx} \label{thm:tradeoff}
For any $k\geq 2$ and $\ell \le k/2$ we can compute a consistent, pairwise independent sample of $k$-subsets from a given $b$-set in expected time \\$O(b^{\lceil k/2 + \ell \rceil} \log \log b) + pb^k)$ and space $O(b^{\lceil (k-2\ell)/4  \rceil} + b)$ for a given sampling probability $p$, such that $1/p = O(b^k)$ and $p$ can be described in one word. 
\end{thmx}

\begin{proof} We need space $O(b)$ to store the $b$-set. Assume that we iterate over $2\ell$-subsets $(a_1, \dots, a_{2\ell})$, without storing them and their hash values. We assume that we have fixed $\ell$  elements among the first $\lfloor k/2 \rfloor$  elements, and $\ell$  elements among the last $\lceil k/2 \rceil$ ones. We compute the value $h^\ell = (h(a_1) + \dots + h(a_\ell) - h(a_{\lfloor k/2 \rfloor + 1}) -\dots - h(a_{\lfloor k/2 \rfloor + \ell}) ) \mbox{ mod } q$. We now want to determine all $(k-2\ell)$-subsets for which the sum of the hash values of the first $\lfloor k/2 \rfloor-\ell$  elements equals the sum of the hash values of the last $\lceil k/2 \rceil - \ell$  elements minus the value $h^\ell$. Essentially, we can sort $(k-2\ell)$-subsets according to their hash value in the same way as before and the only difference is that we subtract $h^\ell$ from the hash value of the last  $\lceil k/2 \rceil - \ell$  elements. Thus, we can use two priority queues, where in the second one we have subtracted $h^\ell$ from the hash value of each $(\lceil k/2 \rceil - \ell)$-subset, output the minima and look up for collisions. Disregarding the space for storing the $b$-set, the outlined modification requires time $O(b^{\lceil k/2 + \ell \rceil} \log \log b)$ and space $O(b^{\lceil (k-2\ell)/4  \rceil})$ to process a given $b$-set. 
\end{proof}

\paragraph{\bf Discussion.}
Let us consider the scalability of our approach to larger values of~$b$, assuming that the time is not dominated by iterating through the sample. If we are given a processor that is $x$ times more powerful, this will allow us to increase the value of $b$ by a factor $x^{1/\lceil k/2 \rceil}$. This will work because the space usage of our approach will only rise by a factor $\sqrt{x}$, and, as already discussed, we expect a factor $x^{0.8}$ more space to be available. An algorithm using space $b^{\lceil k/2 \rceil}$ would likely be space-bounded rather than time-bounded, and thus only be able to increase $b$ by a factor of $x^{0.8/\lceil k/2 \rceil}$. At the other end of the spectrum an algorithm using time $x^k$ and constant space would only be able to increase $b$ by a factor $x^{1/k}$.

\section{Applications of Consistent Subset Sampling}

\subsection{Estimating the number of frequent itemsets} \label{sec:freq_est}

Our original motivation for devising an efficient consistent sampling method was to improve algorithms for the fundamental task of mining frequent itemsets.

We recall again, that in this problem one is given a set of $m$ transactions $T_1, T_2, \dots, T_m$, which are subsets of the ground set of items $\mathcal{I}$.  
A fundamental question is finding the set of frequent $k$-itemsets, that is, subsets of $k$ items occurring in at least $\gamma m$ transactions for a user specified threshold $0<\gamma<1$, see \cite{dmbook} for an overview. The problem is \#P-hard \cite{gunetal} and even hard to approximate \cite{zuckerman}.

Classic algorithms like Apriori~\cite{apriori} and FP-Growth~\cite{FPgrowth} address the problem
and are known to efficiently handle even large input data sets that occur in
practice. However, as shown in~\cite{viara}, their complexity is determined by the number
of frequent itemsets and a lower support threshold may cause exponential
running time. In such cases, one needs to adjust the support threshold and run
the algorithm again with the hope that this time it will run in reasonable time
and will still produce a fair amount of association rules.
Obviously, the above strategy can be very expensive and therefore the basic
frequent set mining algorithms have been further refined to better handle worst
case scenarios~\cite{brinetal,parketal,savasereetal}. However, the number of frequent itemsets is a natural lower bound on the running time of any algorithm computing them.

A recent line of research has started to attack the problem from a different perspective. As Geerts et al.~\cite{geertsetal} artfully describe the goal: {\em ``Indeed, the problem is precisely to predict a combinatorial explosion without suffering from it."} More precisely, randomized algorithms have been designed that accurately predict the number of frequent itemsets without explicitly generating them~\cite{bolgross,sketch_estimate}.

The first randomized algorithm for frequent itemsets mining
appears to be the one by Toivonen~\cite{Toivonen96}. It is based on the natural idea that in a random
sample of the transactions in the database we expect frequent itemsets to be
still frequent. The sample is then used to derive the frequent itemsets in two
passes over the data. However, the running time still depends on the number of
frequent itemsets and thus this approach can be computationally prohibitive.
Also, in a recent work Jin et al.~\cite{sketch_estimate} analyzed the behavior of the above sampling estimator and showed that although it is {\em asymptotically} unbiased, in practice, for samples of reasonable size, it tends to overestimate the number of frequent itemsets. The first randomized algorithm for frequent itemset mining that is not based on sampling of transactions is due to Boley and Grosskreutz~\cite{bolgross}. Their method builds upon a Monte Carlo approach and then uses {\em almost uniform sampling}~\cite{mcmc}. Unfortunately, but not surprisingly, the worst case running time is exponential. Boley and Grosskreutz propose a heuristic procedure that needs a polynomial number of steps and experimentally show that for dense datasets it efficiently estimates the number of frequent itemsets.  However, no theoretical analysis of the approximation guarantee is presented.

A recent work by Amossen et al.~\cite{AmossenCP10} studies the special case of efficient estimation of the number of distinct pairs occurring in a transactional database. (In the frequent itemset mining context this corresponds to estimating the number of 2-itemsets with support threshold at least 1.) 
Building upon a work for estimating the number of distinct elements in a data stream by Bar-Yossef et al.~\cite{baryossef02counting}, their algorithm 
obtains a $(1\pm \varepsilon)$-approximation of the number of distinct pairs with high probability by storing hash values for $O(1/\varepsilon^2)$ pairs. The main contribution of the proposed algorithm is that each transaction can be processed in expected linear time in its length as opposed to the na\"ive application of~\cite{baryossef02counting} which would require quadratic processing time per transaction. The Amossen et al. algorithm can be extended to estimating the number of frequent pairs for a user defined frequency threshold but it is not clear how to adjust it to the problem of approximate frequent $k$-itemset counting.   


Our approach based on consistent subset sampling applies in a natural way to this problem. We simply associate transactions with $b$-sets and $k$-itemsets with $k$-subsets. It is drastically different from previous sampling approaches for the problem~\cite{bolgross,sketch_estimate}, and gives the first nontrivial algorithm for the problem with rigorously understood complexity and error guarantee. Moreover, it can work in a streaming fashion requiring only one pass over the input, thus it has a wider range of applications. 

Note that since sampling is a powerful technique our algorithm can be applied to estimating the number of {\em interesting} $k$-itemsets for any efficiently computable criterion for interestingness. For example, one can estimate the number of $k$-itemsets with items satisfying certain user specified constraints or $k$-itemsets with frequency in a given range. The following results apply to the generalized version of the problem. Also note that the time-space trade-off from Theorem~\ref{thm:tradeoff} applies to the next theorems but for a better readability we will use the time and space complexity from Theorem~\ref{thm:main}. 

In the following theorems we will say that an algorithm returns an {\em $(\varepsilon, \delta)$-approximation} of some quantity $Q$, if it returns a value $\tilde{Q}$, such that $(1- \varepsilon)Q \leq \tilde{Q} \leq (1+\varepsilon)Q$ with probability at least $1 - \delta$ for any $0 < \varepsilon, \delta < 1$.

\begin{figure}  
\renewcommand{\thealgorithm}{}
{\sc SingleRunEstimator}
\begin{algorithmic}[0]
\medskip
\REQUIRE \textbf{Input:} stream of $m$ transactions $\mathcal{T}$ over $n$ items, $2k$-wise independent hash function $h: \mathcal{I} \rightarrow [q]$, double $\alpha, \varepsilon \in (0,1]$
\medskip
\STATE Set $s= \frac{8}{\alpha \varepsilon^2}$
\STATE Let $H$ be a hashtable for $s$ $\langle k$-itemsets, integer$\rangle$-entries
\FOR{each transaction $T \in \mathcal{T}$}
\FOR{each $k$-itemset $I$ output by {\sc ConsistentSubsetSample$(T, h)$}} 
\STATE Update$(H, I)$.
\IF{the number of sampled $k$-itemsets exceeds $32s$}
\RETURN (-1,-1).
\ENDIF
\ENDFOR
\ENDFOR
\IF{the number of sampled $k$-itemsets is less than $s$}
\RETURN (-1,-1).
\ENDIF
\STATE Let $\tilde{\alpha}'$ be the ratio of frequent $k$-itemsets in $H$ and $s'$ the number of sampled $k$-itemsets.
\RETURN{ $(\tilde{\alpha}', \tilde{z})$, $\tilde{z} = qs'$, as estimates of $\alpha'$ and $z$}  
\end{algorithmic}

\bigskip
\renewcommand{\thealgorithm}{}
{\sc GuessingEstimator}
\begin{algorithmic}[0]
\REQUIRE \textbf{Input:} stream of $m$ transactions $\mathcal{T}$ over $n$ items, double $\alpha, \varepsilon \in (0,1]$
\medskip
\STATE Set $s= \frac{8}{\alpha \varepsilon^2}$
\FOR{$i = 0$ to $\log m$}
\STATE Run in parallel
\STATE $z_i = 2^i {b \choose k}$
\STATE choose a $2k$-wise independent $h: \mathcal{I} \rightarrow [q]$ for $q = \text{max}(\frac{z_i}{2s}, 1)$
\STATE $(\tilde{\alpha_i}, \tilde{z_i})=$ {\sc SingleRunEstimator($\mathcal{T}, h, \alpha, \varepsilon$)}.
\ENDFOR
\STATE Let $\tilde{\alpha}_j$ and $\tilde{z}_j$ to be two estimates for some $j \in [\log m]$ different from -1, if existent.
\RETURN{$(\tilde{\alpha}, \tilde{z})$.}
\end{algorithmic}

\bigskip
\renewcommand{\thealgorithm}{}
{\sc FrequentItemsetEstimator}
\begin{algorithmic}[0]
\REQUIRE \textbf{Input:} stream of $m$ transactions $\mathcal{T}$ over $n$ items, double $\alpha, \varepsilon, \delta \in (0,1]$
\medskip
\STATE Run $\log 1/\delta$ copies of {\sc GuessingEstimator}$(\mathcal{T}, \alpha, \varepsilon)$ in parallel.
\STATE Let $\tilde{\alpha}_{m}$ and $\tilde{z}_{m}$ be the medians of the $\log 1/\delta$ estimates $\tilde{\alpha}$ and $\tilde{z}$  
\RETURN{$\tilde{f}=\tilde{\alpha}_m \cdot \tilde{z}_m$}
\end{algorithmic}

\caption{Estimating the number of frequent $k$-itemsets.}
\end{figure}

\begin{thmx} \label{thm:frequent_itemsets_estim}
Let $\mathcal{T}$ be a stream of $m$ transactions of size at most $b$ over a set of $n$ items and $f$ and $z$ be the number of frequent and different $k$-itemsets, $k \ge 2$, in $\mathcal{T}$, respectively. For any $\alpha, \varepsilon, \delta >0$ there exists a randomized algorithm running in expected time $O(mb^{\lceil k/2 \rceil}\log \log b + \frac{\log m \log \delta^{-1}}{\alpha \varepsilon^2})$ and space $O(b^{\lceil k/4 \rceil} + \frac{\log m \log \delta^{-1}}{\alpha\varepsilon^2})$ in one pass over $\mathcal{T}$ returning a value $\tilde{f}$ such that
\begin{itemize}
\item if $f/z \ge \alpha$, then $\tilde{f}$ is an $(\varepsilon, \delta)$-approximation of $f$.
\item otherwise, if $f/z < \alpha$,  then $\tilde{f} \le (1 + \varepsilon)f$ with probability at least $1 - \delta$.
\end{itemize}  
\end{thmx}
\begin{proof}

In the following we will analyze the time and space complexity of {\sc FrequentItemsetEstimator} as well as the quality of the returned estimation.

First note that {\sc ConsistentSubsetSampling} guarantees that all occurrences of a given $k$-itemset $I$ will be sampled, therefore we will compute $I$'s exact frequency in $\mathcal{T}$.

We show how to estimate the values of $\alpha' = f/z$ and $z$ in parallel. At the end we return as an estimate $\tilde{\alpha}'\tilde{z}$ from the computed estimates $\tilde{\alpha}'$ and $\tilde{z}$. 
Denote by $F_k$ the frequent $k$-itemsets and by $Z_k$ all $k$-itemsets occurring in $\mathcal{T}$. Let us denote by $f$ and $z$ the cardinalities of $F_k$ and $Z_k$, respectively.

We first show how to obtain an $(\varepsilon', \delta)$-approximation of $f$ for $\varepsilon'$ that will be defined later.  Assume for now that $z$ is known. (We later show how to remove the assumption.) Consider the case when the fraction of frequent $k$-itemsets is $\alpha' \ge \alpha$. Let us estimate how many $k$-itemsets we need to sample in order to achieve an $(1\pm \varepsilon)$-approximation of $F_k$. Assume we have sampled a set $S$ containing $s$ $k$-itemsets. For each $k$-itemset $I \in S$ we introduce an indicator random $X_I$ denoting whether $I$ is frequent. Let $X = \sum_{I \in S} X_I$. Clearly, we have $\mathbb{E}[X_I] = \alpha'$ and $\mathbb{E}[X] = \alpha' s$. As shown in Lemma~\ref{independence_lemma}, the sampling of two $k$-itemsets is pairwise independent and since the $X_I$ are $\{0,1\}$-valued, for the variance of $X$ it holds $\mathbb{V}[X] \le \mathbb{E}[X]$. By applying Chebyshev's inequality we obtain
$$\Pr[|X-\mathbb{E}[X]| \ge \varepsilon' \alpha' s] \le \frac{\mathbb{V}[X]}{\alpha'^2 s^2} \le  \frac{\mathbb{E}[X]}{\varepsilon'^2\alpha'^2 s^2} = \frac{1}{\varepsilon'^2\alpha' s} \le  \frac{1}{\varepsilon'^2\alpha s}.$$

Thus, we need $s \ge \frac{8}{\alpha \varepsilon'^2}$ to bound the probability that the value returned by {\sc SingleRunEstimator} $\tilde{\alpha}'$ is not an $(1\pm \varepsilon')$-approximation of $\alpha'$ to $1/8$. Analogously, we obtain that if $\alpha' < \alpha$, with probability at least $7/8$ we will return an estimate that is bounded by $(1+\varepsilon')\alpha$. Next we show how to obtain the desired number of samples while guaranteeing the claimed worst case space usage.

Now we show that in at least one run of {\sc GuessingEstimator} we will obtain the required $s \ge \frac{8}{\alpha \varepsilon'^2}$ samples. In the following we assume without loss of generality that $s\ge 32$ and that $\log m$ is integer. In {\sc GuessingEstimator} we run several copies of the algorithm in parallel for different values for $z$. Clearly, ${b \choose k} \le z \le m{b \choose k}$.  Thus we run $\log m$ copies with $z_i = {b \choose k} 2^i$, $i \in \{1,2,\ldots,\log m\}$. For the $i$th copy we choose a sampling probability $p_i = \mbox{min}(\frac{2s}{z_i}, 1)$, i.e., the hash function range in the $i$th run of {\sc SingleRunEstimator} is $q = \mbox{max}(\frac{z_i}{2s}, 1)$.  Without loss of generality we assume that $z_j \le z < z_{j+1}$ for some fixed $1 \le j < \log m$. Consider the $j$th copy of the algorithm. We again introduce an indicator random variable $Y_I$ for each $k$-itemset $I \in \mathcal{T}$ denoting whether it has been sampled and set $Y = \sum_{I \in \mathcal{T}} Y_I$. We have $\mathbb{E}[Y] = \frac{z s}{z_j}$ and thus $2s < \mathbb{E}[Y] \le 4s$. Applying Markov's inequality with $\lambda = 32 s$ we then bound the probability that more than $32s$ $k$-itemsets will be sampled to $1/8$. We apply Chebyshev's inequality with $\lambda = s$ in order to bound the probability that less than $s$ $k$-itemsets will be sampled.  $$\Pr[Y < s] \le \Pr[|Y-\mathbb{E}[Y]| \ge s] \le \frac{\mathbb{V}[Y]}{s^2} \le \frac{\mathbb{E}[Y]}{s^2} \le \frac{4}{s} \le \frac{1}{8}.$$ 

By the union bound we bound the probability that either too many $k$-itemsets have been sampled, or not enough $k$-itemsets have been sampled or an inaccurate estimate of $\alpha'$ is returned to at most $3/8$. 

By running $K$ independent copies of the algorithm in parallel, we expect at least $\frac{5}{8}K$ correct estimates. Thus, we expect that the median will be also an $(1\pm \varepsilon')$-approximation of $\alpha'$. Introducing an indicator random variable for each copy, denoting whether we returned an $(1\pm \varepsilon')$-approximation of $\alpha'$, we can apply Chernoff's inequality from section 2 with $\ell = \log \frac{2}{\delta}$ 
and bound the error probability to $\delta/2$. 

The value $z$ is estimated in parallel in analogous way. For a given sampling probability $p$, we expect $pz$ distinct $k$-itemsets in the sample and one can show concentration around the expected value via Chebyshev's inequality. For the sampling probabilities in {\sc GuessingEstimator} one can show again that with probability at least $5/8$ for at least one run of {\sc GuessingEstimator} we return $\tilde{z}$, an $(1 \pm \varepsilon')$-approximation of $z$. The error probability can be bound to $\delta/2$ by running $\log \frac{2}{\delta}$ copies in parallel, thus by the union bound we obtain error probability of $\delta$ that either $\tilde{\alpha}'$ or $\tilde{z}$ are not $(1\pm \varepsilon')$-approximations. (Note that since we know $z \ge {b \choose k}$ we do not have to consider the case that we do not obtain enough samples.) It is easy to see that choosing $\varepsilon = 2\varepsilon'$, $\tilde{f} = \tilde{\alpha}\cdot \tilde{z}$ is an $(\varepsilon, \delta)$-approximation of $f$.

The sampled $k$-itemsets are stored in a hashtable $H$ of size bounded by $O(s)$, and $H$ is updated in constant time. The time and space complexity follow then immediately from the pseudocode and Theorem~\ref{thm:main}. 
\qed 
\end{proof}

\subsection{Parallelizing frequent itemset mining algorithms}

In a recent work Campagna et al.~\cite{pairse} propose a new method for parallelizing frequent items mining algorithm based on hashing when applied to mining frequently co-occurring item-item pairs in transactional data streams. Let us first recall the two randomized frequent items mining algorithms {\sc Count-Sketch}~\cite{count_sketch} and {\sc CountMin-Sketch}~\cite{count_min}. We are given a stream of items $\mathcal{I}$ and want to estimate the frequency of each item after processing the stream. In both algorithms we maintain a sketch $Q$ of the stream of the items processed so far. $Q$ consists of $q$ buckets. For each incoming item $i \in \mathcal{I}$ we update one of the $q$ buckets in the sketch. The bucket is chosen according to the value $h(i)$ for a pairwise independent hash function $h:\mathcal{I} \rightarrow [q]$. The algorithms differ in the way we update the bucket $i$ is hashed to and achieve an additive approximation for the frequency of each individual item in terms of the 1- and 2-norm of the stream frequency vector, i.e., the vector obtained from the items frequencies. We refer the reader to the original works~\cite{count_sketch,count_min} for a thorough description of the algorithms and the approximation guarantees.

It is clear that {\sc Count-Sketch} and {\sc CountMin-Sketch} can be applied in a trivial way to the estimation of $k$-itemset frequencies in transactional data streams by simply generating all $k$-itemsets occurring in a given transaction and treating them as items in a stream. This is also the only known approach for frequent $k$-itemset mining in transactional streams with rigorously understood complexity and approximation guarantee. Campagna et al.~\cite{pairse} propose a novel approach to parallelizing the above algorithm applied to the estimation of pairs (or 2-itemsets) frequencies in transactional streams. The design of our hash function allows the generalization of the main idea from~\cite{pairse} to $k$-itemset mining for $k > 2$. Assume the sketch $Q$ consists of $q$ buckets and we are given $p$ processors. We want to distribute the processing of $k$-itemsets occurring in a given transaction among the $p$ processors. Assume without loss of generality that $t = q/p$ is integer. We number the processors from 0 to $p-1$ and assign to the $i$th processor the buckets in the sketch numbered between $it$ and $(i+1)t - 1$ for $i \in [p]$. Thus, we want to efficiently find all $k$-itemsets with a hash value in the given range. By a simple extension of the {\sc ConsistentSubsetSampling} algorithm from Section~\ref{sec:algorithm} we are able to efficiently find all $k$-itemsets hashing to a given value $r \in [q]$. Assume for simplicity that $k$ is even. We set $j_+ = 0$ and $j_- =r - j_+$. We work with two priority queues $P_+$ and $P_-$ as described in section~\ref{sec:algorithm}. However, $P_+$ will output $k/2$-itemsets in increasing order according to their hash value and $P_-$ will output $k/2$-itemsets in decreasing order. We output all $k/2$-itemsets with hash value $j_+$ from $P_+$ and all $k/2$-itemsets with hash value $j_-$ from $P_-$ and generate all valid combinations. Clearly, this can be done by the methods from Section~\ref{sec:algorithm}. Also, we can extend the algorithm to output the $k$-itemsets with a hash value in a given range of size $t$. Instead of storing in $P_-$ only $k/2$-itemsets with a given hash value, we store all $k/2$-itemsets with a hash value in the range $[j_-, j_-+t)$ for monotonically increasing $j_-$. Using the approach presented in Lemma~\ref{lem:gener} we can guarantee that the space usage will not exceed $O(b^{k/4})$. The above extends to arbitrary $k>2$ as presented in section~\ref{sec:algorithm}.  This yields the following result:

\begin{thmx} \label{thm:parallel}
There exists a pairwise independent hash function $h: \mathcal{I}^k \rightarrow [q]$ that can be described in $O(k)$ words, such that given a transaction $T$ with $b$ items we can output $T^k_R$, the set of $k$-itemsets in $T$ hashing to a value in a given range $R = [i, i+t)$ for $i \in [q], t \le q$, in expected time $O(|T^k_R| + b^{\lceil k/2 \rceil}\log \log b)$ and space $O(b^{\lceil k/4 \rceil})$. 
\end{thmx}

The above theorem essentially says that in a parallelized setting each processor only needs to read the transaction and then can efficiently decide which $k$-itemsets it is responsible for. Thus, {\sc Count-Sketch} and {\sc CountMin-Sketch} applied to transactional data streams can be considerably sped up when several processors are available or, more practically, on a desktop with a multi-core CPU.

\subsection{Estimating the number of subgraphs for incidence streams}

Let $G= (V,E)$ be a simple directed or undirected graph with $n = |V|$ vertices and $m=|E|$ edges and bounded degree $\Delta$. We say that vertices $u$ and $v$ are {\em neighbors} if there exists an edge $(u, v) \in E$.  For many real-life problems the whole graph does not fit in memory and we can access only given parts of it. A well-known model for handling large graphs is to read the graph as an ``incidence" stream, that is, for each vertex only the vertices adjacent to it are given. A $k$-subset of the neighbors of some vertex is called {\em $k$-adjacencies}, many vertices can have the same $k$-adjacencies.  The ordering of vertex appearances is arbitrary. One is then interested in certain structural properties of the graph. 
Due to a large number of applications of particular interest have been algorithms for approximately counting the number of small subgraphs. The problem of approximately counting triangles and estimating clustering coefficients in streamed graphs has received considerable attention~\cite{becchettietal,buriol_et_al_1,buriol_et_al_2,iran_triangles,tri_soda02}. The approach was extended to counting $(i,j)$-bipartite cliques~\cite{buriol_et_al_2}, counting graph minors of fixed size~\cite{bordino_et_al}. Building upon work by Jowhari and Ghodsi~\cite{iran_triangles}, Manjunath et al.~\cite{cycles_counting} present new algorithms for approximate counting of cycles in the {\em turnstile model} where one is allowed to insert as well as to delete edges. The approach was later extended by Kane et al.~\cite{subgraph_counting} to arbitrary subgraphs of fixed size.

\paragraph{\bf Estimating the number of $k$-cliques.}

The currently fastest algorithm for counting triangles in incidence streams is the wedge sampling approach by Buriol et al.~\cite{buriol_et_al_1}.  The algorithm can be easily generalized to counting $k$-cliques for arbitrary $k \ge 3$. One samples at random a $(k-1)$-star~\footnote{An $\ell$-star is a connected graph on $\ell + 1$ vertices with one internal vertex of degree $\ell$ and $\ell$ leafs of degree 1.} $s$ and then in the next incidence lists we check whether there exists a $(k-1)$-star centered with the same set of vertices as $s$ but centered at  a different vertex. If we find $k-1$ or $k$ such $(k-1)$-stars, then we conclude that the sampled $(k-1)$-star $s$ is completed to a $k$-clique. Thus, if $s$ is indeed part of a $k$-clique, then it needs to be either the first or the second $(k-1)$-star among the $k$ such $(k-1)$-stars that are part of the $k$-clique. Let $K_k$ be the number of $k$-cliques and $S_{k-1}$ the number of $(k-1)$-stars in the graph.  By setting $\gamma = \frac{K_k}{S_{k-1}}$ to be the ratio of $k$-cliques to $(k-1)$-stars in the graph, we thus need $O(\frac{k}{\gamma \varepsilon^2} \log \frac{1}{\delta})$ sampled $(k-1)$-stars in order to obtain an $(\varepsilon, \delta)$-approximation of $K_k$. Since we don't know the exact value of $K_k = O(n^k)$, one needs to run $O(k\log n)$ copies in parallel each guessing the right value (as we did in {\sc GuessingEstimator} in Section~\ref{sec:freq_est}.  The na\"ive algorithm for processing each edge would result in processing time of $O(\frac{k}{\gamma \varepsilon^2} \log \frac{1}{\delta})$ but using the clever techniques from~\cite{buriol_et_al_1} one can obtain an amortized processing time of $O((1 + \frac{m k}{n \gamma \varepsilon^2}\log \frac{1}{\delta})k \log n)$ per edge. 

Using Consistent Subset Sampling instead, we obtain an alternative approach that may need considerably less samples. Assume w.l.o.g. that $k$ is even. We want to sample all subsets of $k$ vertices $(v_1,v_2,\ldots, v_k)$ that appear together in some $(k-1)$-star and it holds $h(v_1,v_2,\ldots,v_{k/2}) = h(v_{k/2+1}, v_{k/2+2},\ldots, v_{k}) \mbox{ mod } q$. (We remind the reader that $h(v_1,\ldots, v_\ell) = h(v_1) + \ldots + h(v_\ell) \mbox{ mod } q$ for a suitably defined hash function $h: V \rightarrow [q]$.) Remember that we consider only lexicographically ordered sets of vertices. Let be given a list of vertices $(u_1, u_2, \ldots, u_\Delta)$ incident to some  $u \in V$. In order to find all $(k-1)$-stars centered at $u$ that satisfy the sampling condition, we distinguish between two kinds of $(k-1)$-stars centered at $u$ depending on whether $u$ is among the first $k/2$ vertices in lexicographic order or not. Given the hash value $h(u)$, we then first look for subsets $W_1$ of $\lfloor (k-1)/2 \rfloor - 1$ vertices and $W_2$ of $\lceil (k-1)/2 \rceil$ vertices such that $h(W_1) + h(u) = h(W_2) \mbox{ mod } q$. This can be done by a straightforward extension of the approach in Section~\ref{sec:algorithm}. Similarly, we next look for subsets $Y_1$ of $\lfloor (k-1)/2 \rfloor$ vertices and $Y_2$ of $\lceil (k-1)/2 \rceil -1$ vertices such that $h(Y_1) = h(Y_2) + h(u) \mbox{ mod } q$. From the above discussion and the results from Theorem~\ref{thm:main} and Theorem~\ref{thm:frequent_itemsets_estim} we then obtain the following result:

\begin{thmx}
Let $G= (V, E)$ be a graph with $n$ vertices, $m$ edges and bounded degree $\Delta$ revealed as a stream of incidence lists. Let further $K_{k}$ be the number of $k$-cliques in $G$ and $S_{k-1}$ the number of $(k-1)$-stars in $G$ for $k \ge 3$. For any $\gamma, \varepsilon, \delta \in (0,1]$ there exits a randomized algorithm running in expected time $O(n\Delta^{\lceil k/2 \rceil} \log \log \Delta + \frac{\log n \log \delta^{-1}}{\gamma \varepsilon^2})$ and space $O(\Delta^{\lceil k/4 \rceil} + \frac{\log n \log \delta^{-1}}{\gamma \varepsilon^2})$ in one pass over the graph returning a value $\tilde{K}_{k}$ such that 
\begin{itemize}
\item if $K_{k}/S_{k-1} \ge \gamma$, $\tilde{K}_{k}$ is an $(\varepsilon, \delta)$-approximation of $K_{k}$.
\item otherwise, if $K_{k}/S_{k-1}  < \gamma$,  $\tilde{K}_{k} \le (1 + \varepsilon)K_{k}$ with probability at least $1 - \delta$.
\end{itemize} 
\end{thmx}

Note that the number of required samples is a factor of $k$ smaller than in the extension of the triangle counting algorithm by Buriol et al. For graphs where the term $\Delta^{\lceil k/4 \rceil}$ is not dominating, our algorithm can thus be considerably more efficient.

\paragraph{\bf Estimating the number of bipartite cliques.}

Consistent Subset Sampling applies to counting the number of 4-cycles by estimating the number of distinct 2-adjacencies occurring at least two times in the incidence stream.

More generally, we can count non-induced $(i, j)$-bipartite cliques, in (directed and undirected) graphs given as an incidence stream. A non-induced bipartite clique, or biclique, in a graph $G = (V, E)$ is defined as $B = (X, Y, E')$ with $X, Y \subset V$, $X \cap Y = \emptyset$ and $(x, y) \in E'$  $\forall x \in X, y \in Y$. Bipartite cliques appear in many areas of computer science, for a list of applications the reader is referred to~\cite{biclique_appl}. 

Let us discuss known algorithms capable of counting $(i, j)$-bipartite cliques and how they compare to our approach. 
The approach in \cite{buriol_et_al_2} is aimed at counting (3,3)-bicliques and is based on a careful extension of the sampling approach for counting triangles from~\cite{buriol_et_al_1}. A direct comparison to their result is difficult and there is a subtle difference between the problems one can handle with sampling and with consistent sampling. 
The algorithm of Buriol et al. samples uniformly at random a $(1,3)$-star and then checks in the remaining stream whether the star is a subgraph of a $(3,3)$-biclique. The authors show this happens with probability $O(\frac{K_{3,3}}{K_{1,3}\Delta^2})$ where $K_{i,j}$ is the number of $(i, j)$-bicliques in $G$ and $\Delta$ is the maximum vertex degree in $G$. The algorithm can be easily generalized to counting $(i,j)$-bicliques. A concrete motivation for the considered problem is the detection of emerging web communities by analyzing the Web graph~\cite{kumaretal}. In this setting vertices refer to Web pages and their neighbors are other Web pages pointed to by them. An $(i,j)$-biclique shows that $i$ pages point to $j$ other pages and this is an indication of a community with common interests. The number of such communities is valuable information about the structure of the Web graph. Note however that popular Web pages like news portals, search engines and discussion forums can be pointed to by many, say $t$, distinct pages. This can cause a combinatorial explosion in the number of $(i, j)$-bicliques for larger $t$. One possible solution is to exclude such popular pages from consideration but this may result in incomplete information. We propose another solution, namely counting bipartite cliques with fixed right-hand size and variable left-hand size, or more precisely to count $(i^+, j)$-bicliques with at least $i$ vertices on the left size and exactly $j$ vertices on the right-hand side. Note that this problem cannot be handled by the sampling algorithm in \cite{buriol_et_al_2} since it will be biased towards sampling frequently occurring $j$-adjacencies. The more general algorithms from~\cite{bordino_et_al,subgraph_counting} also do not apply to the problem of counting $(i^+, j)$-bicliques.

Using consistent sampling we can handle this problem. The time and space bounds and accuracy of estimates are exactly the same as for frequent itemset mining. Let $\Delta$ be the maximum degree of the streamed graph $G$. We consider the incidence list of each vertex as a $\Delta$-set and vertices, we are then interested in $j$-subsets occurring in at least $i$ $\Delta$-sets. The $j$-subsets correspond to $j$-adjacencies. Thus, we obtain the following result.

\begin{thmx} \label{biclique_thm}
Let $G= (V, E)$ be a graph with $n$ vertices, $m$ edges and bounded degree $\Delta$ revealed as a stream of incidence lists. Let further $K_{i^+,j}$ be the number of $(i^+, j)$-bicliques in $G$ and $A_{j}$ the number of $j$-adjacencies in $G$ for $i \ge 1, j \ge 2$. For any $\gamma, \varepsilon, \delta \in (0,1]$ there exits a randomized algorithm running in expected time $O(n\Delta^{\lceil j/2 \rceil} \log \log \Delta + \frac{\log n \log \delta^{-1}}{\gamma \varepsilon^2})$ and space $O(\Delta^{\lceil j/4 \rceil} + \frac{\log n \log \delta^{-1}}{\gamma \varepsilon^2})$ in one pass over the graph returning a value $\tilde{K}_{i^+,j}$ such that 
\begin{itemize}
\item if $K_{i^+,j}/A_{j} \ge \gamma$, $\tilde{K}_{i^+,j}$ is an $(\varepsilon, \delta)$-approximation of $K_{i^+,j}$.
\item otherwise, if $K_{i^+,j}/A_{j}  < \gamma$,  $\tilde{K}_{i^+,j} \le (1 + \varepsilon)K_{i^+,j}$ with probability at least $1 - \delta$.
\end{itemize} 
\end{thmx}

\section{Conclusions}

Finally, we make a few remarks about possible improvements in the running time of our consistent sampling technique. As one can see, the algorithmic core of our approach is closely related to the $d$-SUM problem where one is given an array of $n$ integers and the question is to find $d$ integers that sum up to 0.  
The best known randomized algorithm for 3-SUM runs in time $O(n^2 (\log \log n)^2/\log^2 n)$\cite{baranetal}, thus it is difficult to hope to design an algorithm enumerating all 3-subsets satisfying the sampling condition much faster than in $O(b^2)$ steps.  Moreover, P\v{a}tra\c{s}cu and Williams \cite{patrwill} showed that solving $d$-SUM in time $n^{o(d)}$ would imply an algorithm for the 3-SAT problem running in time $O(2^{o(n)})$ contradicting the {\em exponential time hypothesis} ~\cite{eth}. 
It is even an open problem whether one can solve $d$-SUM in time $O(n^{\lceil d/2 \rceil - \alpha})$ for $d \geq 3$ and some constant $\alpha > 0$ ~\cite{woeg_survey}.

In a recent work Dinur et al.~\cite{dinur_et_al} presented a new ``dissection" technique for achieving a better time-space trade-off for the computational complexity of various problems. Using the approach from~\cite{dinur_et_al}, we can slightly improve the results from Theorem~\ref{thm:tradeoff}. However, the details are beyond the scope of the present paper. 

%
%
%
%
%



\end{document}